\documentclass[conference,letterpaper]{IEEEtran}
\IEEEoverridecommandlockouts
\usepackage[utf8]{inputenc} 
\usepackage[T1]{fontenc}
\usepackage{url}
\usepackage{ifthen}
\usepackage{mathtools}

\usepackage{graphicx}
\usepackage{amssymb, amsmath,amsthm, amsfonts}
\usepackage{tikz}
\usepackage{verbatim}
\usepackage{pifont}
\usepackage{thmtools,thm-restate}

\usepackage[numbers,sort&compress]{natbib}

\newtheorem{theorem}{Theorem}

\newtheorem{definition}[theorem]{Definition}

\newcommand{\be}{\begin{equation}}
\newcommand{\ee}{\end{equation}}
\newcommand{\ben}{\begin{equation*}}
\newcommand{\een}{\end{equation*}}
\newcommand{\ba}{\begin{eqnarray}}
\newcommand{\ea}{\end{eqnarray}}

\newcommand{\suppress}[1]{}

\allowdisplaybreaks

\begin{document}

\renewcommand{\L}{{\mathcal L}}

\newcommand{\1}{{\bf 1}}
\newcommand{\Z}{{\mathds Z}}
\newcommand{\dis}{{\mathsf{dis}} \,}
\newcommand{\ep}{\epsilon}
\newcommand{\vep}{\varepsilon}

\newcommand{\bA}{\mathsf{A}}
\newcommand{\bC}{\mathsf{C}}
\newcommand{\bG}{\mathsf{G}}
\newcommand{\bT}{\mathsf{T}}

\newcommand{\B}{{\mathcal B}}
\newcommand{\C}{{\mathcal C}}

\newcommand{\G}{{\mathcal G}}
\renewcommand{\H}{{\mathcal H}}
\newcommand{\D}{{\mathcal D}}

\newcommand{\Dr}{\mathcal{DR}}
\newcommand{\dof}{\mathbf{D}}
\newcommand{\Dreg}{\dof}

\newcommand{\V}{{\mathcal V}}
\renewcommand{\S}{{\mathcal S}}
\newcommand{\M}{{\mathcal M}}
\newcommand{\N}{{\mathcal N}}
\newcommand{\IN}{{\mathbb N}}
\newcommand{\R}{{\mathbb R}}
\newcommand{\Rs}{{\mathcal R}}
\newcommand{\Os}{{\mathcal O}}
\newcommand{\Ps}{{\mathcal P}}
\newcommand{\K}{{\mathcal K}}
\newcommand{\W}{{\mathcal W}}
\newcommand{\vX}{{\vec{X}}}
\newcommand{\vY}{{\vec{Y}}}
\newcommand{\F}{{\mathbb F}}
\newcommand{\dE}{D_\Sigma}
\newcommand{\q}[2]{Q_{s_{#1},d_{#2}}}
\newcommand{\p}[2]{P_{s_{#1},d_{#2}}}
\newcommand{\m}[2]{M_{s_{#1},d_{#2}}}
\newcommand{\ttt}{3 \times 3 \times 3}
\newcommand{\kkk}{K \times K \times K}
\newcommand{\kk}[1]{#1 \times #1 \times #1}
\newcommand{\cT}{{\cal T}}
\newcommand{\cR}{{\cal R}}
\newcommand{\cN}{{\cal N}}
\newcommand{\cC}{{\cal C}}

\newcommand{\s}{{\bf s}}
\newcommand{\bs}{{\bf s}}
\newcommand{\bc}{{\bf c}}

\newcommand{\EP}{{\rm EP}}

\newcommand{\setx}{\{ x_{(i)}^{K} \}_M }
\newcommand{\setxM}[1]{\{ x_{(i)}^{K} \}_{#1} }

\newcommand{\setX}{\{ X_{(i)}^{K} \}_M }
\newcommand{\setXM}[1]{\{ X_{(i)}^{K} \}_{#1} }

\newcommand{\sety}{\{ y_{(i)}^{K} \}_N }
\newcommand{\setyN}[1]{\{ y_{(i)}^{K} \}_{#1} }

\newcommand{\setY}{\{ Y_{(i)}^{K} \}_N }
\newcommand{\setYN}[1]{\{ Y_{(i)}^{K} \}_{#1} }

\newcommand{\bp}{{\bf p}}
\renewcommand{\r}{{\bf r}}
\newcommand{\x}{{\bf x}}
\newcommand{\y}{{\bf y}}
\newcommand{\z}{{\bf z}}

\newcommand{\Cunc}{C_\text{unc}}

\newcounter{numcount}
\setcounter{numcount}{1}

\newcommand{\eqnum}{\stackrel{(\roman{numcount})}{=}\stepcounter{numcount}}
\newcommand{\leqnum}{\stackrel{(\roman{numcount})}{\leq\;}\stepcounter{numcount}}
\newcommand{\geqnum}{\stackrel{(\roman{numcount})}{\geq\;}\stepcounter{numcount}}
\newcommand{\cnt}{$(\roman{numcount})$ \stepcounter{numcount}}
\newcommand{\rescnt}{\setcounter{numcount}{1}}

\newcommand{\cov}{{\rm \text{coverage}}}
\newcommand{\recost}{{\rm \text{reordering~cost}}}

\newcommand{\Bernoulli}{{\rm Bernoulli}}
\newcommand{\Geometric}{{\rm Geometric}}
\newcommand{\TPC}{{\rm TPC-LP}}
\newcommand{\U}{\overline{U}}
\newcommand{\qc}{\tilde{q}}

\newif\iflong
\longtrue

\newif\ifdraft
\drafttrue

\newcommand{\iscomment}[1]{
\ifdraft
{\color{blue} \bf{{{{IS --- #1}}}}}
\else
\fi}
\newcommand{\adicomment}[1]{
\ifdraft
{\color{red} \bf{{{{ANR --- #1}}}}}
\else
\fi
}

\title{Fundamental Limits of Non-Adaptive Group Testing With Markovian Correlation} 

\author{%
  \IEEEauthorblockN{Aditya~Narayan~Ravi}
  \IEEEauthorblockA{
                    University of Illinois, Urbana-Champaign\\
                    anravi2@illinois.edu}
  \and
  \IEEEauthorblockN{Ilan Shomorony}
  \IEEEauthorblockA{
                    University of Illinois, Urbana-Champaign\\
                    ilans@illinois.edu}
}

\maketitle

\begin{abstract}
THIS PAPER IS ELIGIBLE FOR THE
STUDENT PAPER AWARD. 

We study a correlated group testing model where items are infected according to a Markov chain, 
which creates bursty infection patterns.
Focusing on a very sparse infections regime, we propose a  non-adaptive testing strategy with an efficient decoding scheme that is nearly optimal. 
Specifically, it achieves asymptotically vanishing error with a number of tests that is within a \(1/\ln(2) \approx 1.44\) multiplicative factor of the fundamental entropy bound---a result that parallels the independent group testing setting. 
We show that the number of tests reduces with an increase in the expected burst length of infected items, quantifying the advantage of exploiting correlation in test design.
\end{abstract}

\section{Introduction}

Consider a \emph{correlated group testing} problem where the infection status of an item is correlated with the infection status of the item before it. We are interested in non-adaptive group testing strategies~\cite{balding1996comparative} to detect the infected items. 

Classical group testing, introduced by Dorfman~\cite{dorfman1943detection}, often employs the combinatorial model, a widely used framework. 
In this model, a population of size \(n\) contains exactly \(k\) infected individuals, chosen uniformly at random from all possible subsets of size \(k\).
An alternative approach is the probabilistic i.i.d.~model, which assumes that each individual is independently infected with probability \(q_n\). 
Importantly, these two models are largely equivalent \cite{aldridge2019group}. 
This framework has been extensively explored in the literature~\cite{coja2020information, scarlett2017phase, scarlett2018near, coja2020optimal, aldridge2019group}. 
In this setting, when infections are sparse, group testing can dramatically reduce the number of tests required.

However, this assumption of independence is rarely reflective of real-world infection dynamics. 
Recent studies have begun to acknowledge this limitation and have explored the potential of leveraging known community structures \cite{ahn2021adaptive, nikolopoulos2021group, nikolopoulos2023community, ahn2023noisy, bertolotti2020network, arasli2023group} to enhance group testing efficiency.
Moreover, in many practical scenarios, infections are not uniformly distributed but tend to occur in bursts \cite{colbourn1999group}.
For example, people living in the same living space or neighborhood are more likely to get infected and tested together. 
Such bursty infection patterns are naturally captured by correlated infection models.

Another important motivation for our work arises from imaging-based spatial transcriptomics technologies such as MERFISH~\cite{moffitt2016rna} and Xenium~\cite{janesick2023high}. 
These technologies use 
barcoded ``probes'' to simultaneously test the expression of a group of genes. 
If any of the genes tested are being expressed in a certain location on a tissue, that location lights up when an image is captured under a fluorescence microscope.
The design of these probes can be cast as a group testing problem, where we can target multiple genes during each testing/imaging round.
 Furthermore, the genes that are tested often exhibit strong local dependencies, where genes involved in the same regulatory pathways or located near each other on a chromosome tend to activate together \cite{allocco2004quantifying}, once again resulting in bursty patterns. 
 This ``co-regulation'' has been exploited previously  \cite{cleary2021compressed}, with correlated genes being measured together to reduce the number of tests required. 
Understanding correlated group testing can thus provide insights into how an optimal design 
 of these probes could take advantage of gene co-regulation, reducing the number of testing rounds.
Note that once the probes are designed and introduced in spatial transcriptomics experiments, modifying or adapting the experimental setup is infeasible. This further motivates the need for non-adaptive group testing strategies for this setting.

Delving into the problem more formally, consider a set of \(n\) distinct items labeled \(\{1, 2, \dots, n\}\). 
Each item \(i\) is assigned a binary random variable \(U_i \in \{0,1\}\), where \(U_i = 1\) indicates that item \(i\) is infected, and \(U_i = 0\) otherwise. 
The collective infection status of all items is represented by the \emph{infection vector} \(U^n = (U_1, U_2, \dots, U_{n})\).
The objective is to estimate \(U^n\) using the outcomes of \(T\) non-adaptive group tests. 
Each test simultaneously tests a subset of items. 
A test returns a positive result if at least one infected item is included in the group; otherwise, the result is negative. 

To build intuition, consider the classical independent group testing problem, where each item in a population of size \(n\) is independently infected with probability  \(q_n\). 
This problem is categorized into three regimes \cite{aldridge2019group} based on the expected number of infected items \(nq_n\). In the \emph{very sparse regime} \cite{aldridge2019group, scarlett2018near}, the expected number scales slower than any polynomial in \(n\)  (e.g.,  \( nq_n \sim O(\log n)\) or $O(1)$). 
In the \emph{sparse regime} \cite{coja2020optimal, scarlett2017phase}, it grows sub-linearly, but faster than logarithmic rates (e.g., \(nq_n \sim O(n^\theta), \theta < 1\)). 
In the \emph{dense regime}, the number of infected items scales linearly with \(n\), i.e.,  $q_n \in (0,1)$ is a constant.
Group testing is known to provide significant gains in the very sparse and sparse regimes, where the number of infections are relatively low \cite{aldridge2019group}.

Consider a very sparse regime,
where \(U_i\)'s are independent and distributed as \(\text{Bernoulli}(q_n)\), where \(q_n = k \log (n)/n\) for some constant \(k > 0\). 
Now, consider testing strategies that employ 
\emph{per-item} decoding  (which is sometimes referred to as separate decoding of items~\cite{aldridge2019group}). 
In this approach, the infection status of each item is determined solely based on the information of which tests contain that item and the test outcomes.
In contrast, joint decoding \cite{aldridge2019group, scarlett2017phase, coja2020information} considers the infection status of all items simultaneously, leveraging the entire set of items tested in each test to jointly decode the infection status of the items.
It can be shown that a \emph{per-item} testing and decoding scheme \cite{malyutov1980planning, scarlett2018near} can achieve a vanishing probability of error as \(n \to \infty\) if the number of tests \(T(n)\) satisfies
\begin{align} \label{eq:iidbound}
    T(n) > \frac{nq_n\log{n}}{\ln{2}}.
\end{align}
To achieve this bound, the testing matrix (which specifies which items are included in each test) can be designed in a random fashion, by including each item in each test independently, with a probability \(p = \nu/nq_n\) \cite{scarlett2018near}, where \(\nu\) is a constant.
Moreover \(\nu\) is optimized and set to \(\nu = \ln(2)\) to obtain the bound.
This is followed by Maximum Likelihood (ML) decoding, where each item is decoded independently based on the ML rule. 
Surprisingly, this simple scheme is optimal (under \emph{per-item} decoding) in the independent case \cite{scarlett2018near}.


In contrast, the correlated group testing problem presents a significantly greater challenge compared to the independent case. 
Independent testing strategies become suboptimal in this setting because they ignore the information the infection status of an item provides about other correlated items. 
As a result, applying independent testing in correlated settings fails to reduce the number of tests. 
Prior work~\cite{ahn2021adaptive, ahn2023noisy, nikolopoulos2021group, nikolopoulos2023community} has introduced well-designed \emph{adaptive} strategies that effectively exploit these correlations to reduce the number of required tests. 
However, developing non-adaptive strategies that achieve similar improvements remains difficult. 
For example, \cite{nikolopoulos2023community} proposes a non-adaptive method that reduces the number of tests but only when a false positive error is non-vanishing, 
highlighting the complexity of designing such strategies for correlated cases.

To make progress in the non-adaptive correlated setting,
wWe consider the case where the infection vector \(U^n\) is generated by a two-state Markov chain, 
shown in Figure~\ref{fig:grouptesting}.
For relevant parameter choices, this Markov chain induces bursty infection patterns.
Specifically, define transition probabilities \(\alpha_n = k^{\prime}\log(n)/n \) and \(\beta < 1\) (a constant). 
The parameter \(\beta \) is the probability  of the infection process leaving a burst of infected items and the expected burst length is \(1/\beta\).
The marginal infection probability remains \(q_n = k \log(n)/n\), consistent with the very sparse regime.

\begin{figure}[ht]
\vspace{-3mm}
\centering
\includegraphics[width=0.25\textwidth]{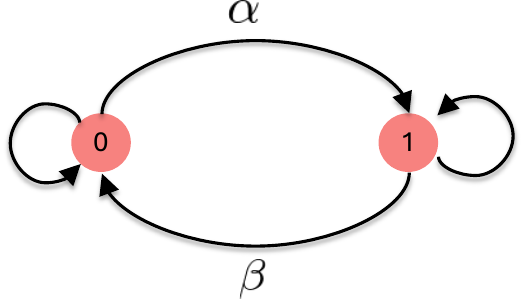}
\caption{
 The \(2\)-state Markov chain that generates the infection vector \(U^n\) }\label{fig:grouptesting}
\end{figure}

In this Markovian setting, we design a near optimal non-adaptive \emph{per-item} testing and decoding strategy, which explicitly leverages the underlying correlation structure to minimize the number of tests. 
We show that this correlation-aware design achieves a vanishing probability of error when the number of tests satisfies
\begin{align}
    T(n) > \frac{\beta\, nq_n\log{n}}{\ln{2}}.
\end{align}
The key difference with respect to (\ref{eq:iidbound}) is the factor \(\beta\), reflecting the advantage of leveraging correlation between items. 
Specifically, the number of tests in the Markovian setting is inversely proportional to the expected burst length.




Our approach utilizes a two-stage approach to create the (non-adaptive) test matrix: a first stage that selects blocks to be included in each test, and a second stage that randomly samples items from the selected blocks to be included in the tests.
Central to this strategy is our efficient  \emph{per-item} decoder, which operates by first performing an initial screening that considers whether a given item is part of a non-infected tested group.
If it is not, we follow up by applying a threshold on the number of times that item is tested to estimate its infection status. 
Surprisingly, this simple decoding rule leverages the testing information and provably reduces the number of tests required, achieving near-optimal performance with low decoding complexity.

Moreover, in the Markovian setting, a fundamental general result in group testing \cite{wolf1985born} implies that for any group testing strategy and any decoding strategy achieving a vanishing probability of error (including joint decoding schemes), the number of tests \(T\) must satisfy the fundamental lower bound
\begin{align}
    T \ge H(U^n) = \beta nq_n \log{n} + o(nq_n \log{n}),
\end{align}
where \(H(\cdot)\) is the entropy.
Our proposed testing and decoding scheme achieves close to this bound, requiring only a multiplicative factor of \(1/\ln{2}\) more tests. 
Similar to the independent case, this demonstrates that our strategy is not only simple and efficient but also near-optimal, as it achieves a test count within a constant factor of the theoretical minimum. 
Our main contributions are the following:

\begin{itemize}
\item We introduce a correlated group testing setting, where correlations between items are Markovian.

    \item We introduce a novel non-adaptive randomized block testing strategy tailored to the Markovian setting.

    \item We introduce a simple \emph{per-item} thresholding-based decoding strategy that is computationally efficient.

    \item We prove that the testing and per-item decoding strategy achieves within a multiplicative factor of \(1/\ln{2}\) of the information-theoretic lower bound.

    

\end{itemize}

\noindent \textbf{Related Work:}
Group testing has been extensively studied under uniform or i.i.d.~priors \cite{aldridge2019group, scarlett2017phase, scarlett2018near, colbourn1999group, berger2002asymptotic}, as well as in heterogeneous populations with non-identical infection probabilities \cite{hwang1975generalized, li2014group, kealy2014capacity}. 
Several works have extended classical models to incorporate infection bursts~\cite{colbourn1999group, li2023finding, lin2012synthetic, muller2004consecutive, bui2021improved}, typically by imposing deterministic constraints on factors like burst length and number of bursts.
Our proposed setting can be seen as a probabilistic alternative to these burst models.



Correlations between the infection status of items have also been introduced in the context of community-aware group testing.
This setting leverages population structure by assuming the individuals can be partitioned into disjoint families or by employing stochastic block models with correlation-driven infection patterns~\cite{nikolopoulos2021group, nikolopoulos2023community, ahn2021adaptive, ahn2023noisy}.
Related approaches, such as those in \cite{bertolotti2020network, arasli2023group, lau2022model}, incorporate subgroup structures, while our approach utilizes a Markovian framework that operates without predefined partitions.
Network structure is also present in graph-constrained group testing but, in this setting, the network imposes constraints on the test design~\cite{harvey2007non, cheraghchi2012graph, karbasi2012sequential}.


Adaptive strategies for group testing are explored in \cite{ahn2021adaptive, ahn2023noisy, nikolopoulos2021group, nikolopoulos2023community}. 
Notably, non-adaptive methods such as the one proposed in \cite{nikolopoulos2023community} demonstrate reductions in the number of required tests, albeit with non-vanishing false positive rates.




\section{Problem Setting}\label{sec:ProblemSetting}

We consider a correlated group testing framework where \(n\) distinct items, labeled \(\{1, 2, \dots, n\}\), are each assigned a binary label \(U_i \in \{0, 1\}\). 
We impose a Markovian structure on \(P(U^n)\), specifically a 2-state Markov chain with state space \(S = \{0, 1\}\) and transition probabilities \(p_{0,1} = \alpha\) and \(p_{1,0} = \beta\) as in Figure~\ref{fig:grouptesting}. We generate the infection random vector \(U^n\) by initializing the Markov chain in its steady state 
\[
(q, 1-q) := \left(\frac{\alpha}{\alpha + \beta}, \frac{\beta}{\alpha + \beta}\right),
\]
and then recording \(n\) consecutive states.

Let \(T := T(n)\) be the number of tests to be conducted. For fixed \(n\) and \(T\), a testing strategy is encoded by a matrix \(\mathsf{X} \in \{0,1\}^{T \times n}\). Specifically, \(X_{t,i} = 1\) if item \(i\) is included in test \(t\), and \(X_{t,i} = 0\) otherwise.
We observe the test outcomes in a vector \(Y = [Y_1,\dots,Y_T]\), where each entry \(Y_t\) follows
\[
    Y_{t} = 
    \begin{cases}
       1 & \text{if }\exists\, i \in \{1,\dots,n\} : U_i = 1 \text{ and } X_{t,i} = 1,\\
       0 & \text{otherwise}.
    \end{cases}
\]
Equivalently, 
\[
Y_t = \mathbf{1} \left\{\sum_{i=1}^n U_i X_{t,i} > 0\right\},
\] 
with \(\mathbf{1}\{\cdot\}\) denoting the indicator function. A decoding rule then maps \(\{0,1\}^{T \times n} \times \{0,1\}^{T}\) to an estimate \(\hat{U}^n \in \{0,1\}^n\). Let 
\(
\mathcal{E} = \{U^n \neq \hat{U}^n\}
\)
denote the error event.



Our goal is to design a sequence of testing matrices and decoding rules for which $\Pr(\mathcal{E}) \to 0$ as $n \to \infty$.
As suggested by the i.i.d.~result in (\ref{eq:iidbound}), $T(n)$ will need to scale as $n q_n \log n$. 
In the Markov setting that we consider, $q_n = \alpha_n/(\alpha_n+\beta)$ is the stationary probability of an infection.
Based on this scaling, we define an achievable testing rate as follows:

\begin{definition}
A testing rate \( \tau\) is achievable if there exists a sequence of test matrices 
\(\mathsf{X}_n \in \{0,1\}^{T(n) \times n}\) and corresponding decoding rules $g_n$ such that
\[\tau = \lim_{n\to \infty} \frac{T(n)}{n q_n \log n}\]
and $\Pr(\mathcal{E}) \to 0$ as $n \to \infty$.
Moreover, we let \(\tau^{*}\) be the infimum of  all achievable rates $\tau$ and refer to it as the minimal group testing rate.
\end{definition}






Following the terminology in the group testing literature, we focus on the \emph{very sparse regime}~\cite{aldridge2019group}, 
which corresponds to
\begin{align}\label{eq:scaling}
    nq_n = o\bigl(n^\theta\bigr) \quad \text{for every } \theta > 0.
\end{align}
This condition includes all growth rates strictly slower than any sub linear function of \(n\), such as logarithmic or iterated-logarithmic functions. 
For concreteness, we restrict our analysis to the case where \(\alpha_n\) is defined as
\begin{align}
    \alpha_n =  \frac{k^{\prime}\log n}{n},
\end{align}
for some constant \(k^{\prime}\). 
We can then compute \(q_n\) as
\begin{align}
    q_n = \frac{\alpha_n}{\alpha_n + \beta} = \frac{k\log{n}}{n} + o\left(\frac{\log{n}}{n}\right),
\end{align}
where 
\(k := k^{\prime}/ \beta\). 
Notice that the expected number of defectives is \(nq_n \sim O(\log{n})\). 
For simplicity, we omit lower-order terms of \(o(\log n/n)\) that are added to \(q_n\) and assume that
\[
 q_n =  \frac{k\log{n}}{n},
\]
in the rest of this work.
However, our results remain valid even when the lower-order terms are included. 


Although a natural next step is to generalize our approach to denser regimes (i.e., \(nq_n \sim O(n^\alpha)\)), a rigorous proof of such extensions is left for future work.

\noindent \textbf{Notation:} Throughout the paper, \( \log(\cdot) \) represents the logarithm in base \( 2 \). For functions \( a(n) \) and \( b(n) \), we say \( a(n) = o(b(n)) \) if \( a(n)/b(n) \to 0 \) as \( n \to \infty \). Similarly, we say \( a(n) = O(b(n)) \) if there exists a constant \( C > 0 \) and an integer \( n_0 \) such that 
$|a(n)| \leq C |b(n)|$  for all $n \geq n_0$.

\section{Main Results}\label{sec:MainResults}

Our main results
include a converse bound (adapted from prior work) and an achievability result for our proposed strategy.
The following result, adapted from~\cite{wolf1985born},
establishes a fundamental lower bound on the number of tests required for any group testing strategy with a vanishing error probability.

\begin{theorem}[Converse]
\label{thm:converse}
Under the Markovian correlation model, the minimal group testing rate \(\tau^{*}\) satisfies
\begin{align}
    \tau^{*} \geq \beta,
\end{align}
where \(\beta\) 
is the inverse of the expected infection burst length. 
\end{theorem}

The proof of Theorem~\ref{thm:converse} is
is based on an entropy bound and is considered
in the longer version of this paper~\cite{markovian_group_testing}. 

In the i.i.d.~setting, where the components of \(U^n\) are independent, the corresponding lower bound is
\begin{align}
   \tau^* \geq 1,
\end{align}
as shown in prior work \cite{wolf1985born}. Comparing these bounds demonstrates that the Markovian correlation structure should reduce the number of tests by a constant factor \(\beta\). 
Intuitively, this suggests that tests can be designed that also reduce the number of tests by this factor. 
Our main result establishes this.

\begin{theorem}[Achievability]
\label{thm:achievability}
Under the Markovian correlation model, the minimal group testing rate \(\tau^*\) satisfies
\begin{align}
    \tau^* \leq \frac{\beta}{\ln 2} \approx 1.44\beta.
\end{align}
\end{theorem}


This bound is achievable using a non-adaptive randomized block testing design and a per-item decoding rule. 
A \textit{per-item decoder} determines whether item \( i \) is infected based solely on the \( i \)-th column of the test matrix \( \mathsf{X} \) and the test outcome vector \( Y \).
The additional factor \(1 / \ln 2\) reflects the inherent inefficiency of per-item decoding compared to joint decoding.

Comparing Theorems~\ref{thm:converse} and \ref{thm:achievability}, we see that our testing strategy achieves a number of tests \(T(n)\) that is within a factor of \(1 / \ln 2\) of the fundamental limit. 
This parallels a similar result that exists for per-item decoders in the i.i.d. case \cite{scarlett2018near}.
This gap arises from the use of per-item decoding, which, while computationally efficient, sacrifices some optimality compared to joint decoding strategies. 
Nonetheless, the proposed strategy significantly improves upon naive designs that ignore the Markovian correlation structure. 

\section{Proof of Theorem~\ref{thm:achievability}}\label{sec:proof}

\subsection{Testing and Decoding Scheme}
To exploit the correlation among items, we adopt a two-stage testing procedure as shown in Figure~\ref{fig:testing}.
Partition the \(n\) items into equally sized contiguous blocks \(\{B_1, B_2, \dots, B_{n/C}\}\) with \(C\) items each. 
Here \(C\) is a fixed, large constant.
        Specifically, 
    \begin{align*}
        B_i := \{(i-1)C+1,(i-1)n/C+2, \dots, iC\}.
    \end{align*}
We adopt a two stage test. First each block \(B_i\)  is selected with probability \(p_1\). 
We then select each item within the selected blocks independently with probability \(p_2\).

More formally, each row of the test matrix $\mathsf{X}$ is independently generated as follows.
For each block \(B_\ell\), \(\ell = 1, \dots, n/C\), 
we draw independent \(\text{Bern}(p_1)\) random variables \(W_l\), which determine whether the block is selected during the coarse selection phase. 
For each item \(j = 1, \dots, n\), we also draw independent \(\text{Bern}(p_2)\) random variables \(Z_j\), which determine whether an item within a selected block is tested during the fine selection phase. 
The entries on the $i$th row of the test matrix \(X\) are then
\begin{align}
    X_{ij} = W_\ell Z_j, \quad \text{for } j \in B_\ell.
\end{align}
In this way, \(X_{ij} = 1\) if and only if block \(B_\ell\) containing item \(j\) is selected in the first selection phase (\(W_t = 1\)) and item \(j\) is independently selected in the fine selection phase (\(Z_j = 1\)).
 
\begin{figure}[ht]
\centering
\includegraphics[width=0.5\textwidth]{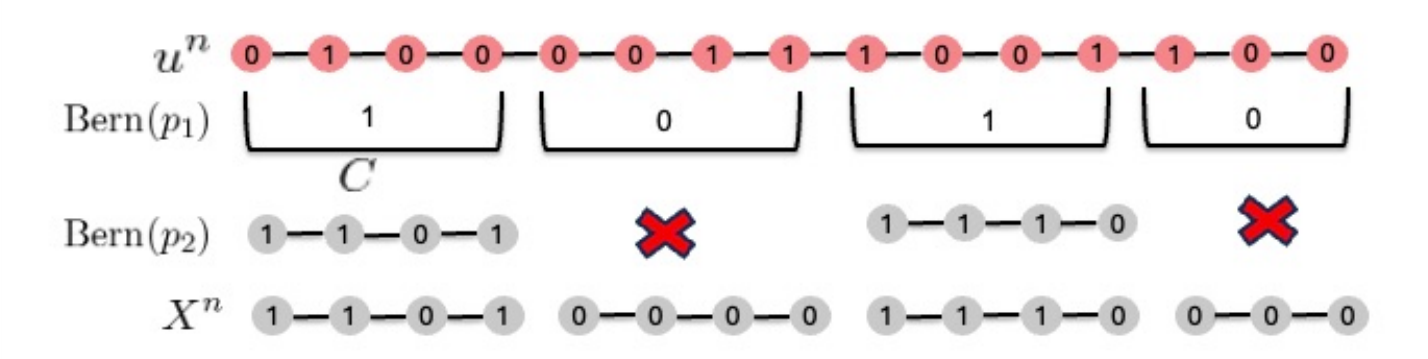}
\caption{
 Grouped testing which selects groups with Bern\((p_1)\). Within selected groups, each item is selected with Ber\((p_2)\).}\label{fig:testing}
\end{figure}


We employ a \emph{per-item} decoding rule that determines the infection status of each item \(i\) using two indicators: an intermediate estimate \(\tilde{u}_i\) to flag potential infections and a final estimate \(\hat{u}_i\). 

The decoding process consists of two steps:
For each item \(i\), if there exists any test \(t\) where the item was included (\(X_{t,i} = 1\)) and the test result was negative (\(Y_t = 0\)), the item is not infected. 
In this case, set \(\tilde{u}_i = 0\). 
If no such test exists, set \(\tilde{u}_i = 1\), i.e.,
\[
    \tilde{u}_i = \begin{cases}
    0, & \text{if } \exists\, t \text{ such that } (X_{t,i} = 1, Y_t = 0) \\
    1, & \text{otherwise.}
    \end{cases}
\]
If \(\tilde{u}_i = 0\), the item is non-infected, so set \(\hat{u}_i = 0\). 
If \(\tilde{u}_i = 1\), compute the total number of tests in which item \(i\) participates:
\[
    X_i := \sum_{t=1}^T X_{t,i}.
\]
Compare \(X_i\) to a threshold \(\gamma = p(1-\varepsilon)\cdot T\), where \(p := p_1 p_2\). 
This choice of \(\gamma\) will be explained in the error analysis section below.
If \(X_i \geq \gamma\), set \(\hat{u}_i = 1\), otherwise set \(\hat{u}_i = 0\), i.e.
\[
    \hat{u}_i = \begin{cases}
    1, & \text{if } X_i \ge \gamma \\
    0, & \text{otherwise.}
    \end{cases}
\]
We now proceed to analyzing the probability of error of the above testing and decoding rule. 

\subsection{Error Analysis}
Let \(\mathcal{E}_i\) denote the event that a single item \(i\) is misclassified; that is, a \emph{false positive} or \emph{false negative} event occurs. 
\begin{align}\label{eq:unionbound}
    \Pr(\mathcal{E})
    =
    \Pr\left(\bigcup_{i=1}^n \mathcal{E}_i\right)
    \stackrel{(a)}{\le}
    \sum_{i=1}^n \Pr(\mathcal{E}_i)
    \stackrel{(b)}{\le}
    n\,\Pr\left(\mathcal{E}_{\max}\right),
\end{align}
The step in (a) follows from the union bound, while in (b) we define \(\Pr(\mathcal{E}_{\max}) = \max_{i} \Pr(\mathcal{E}_{i})\). 

To analyze the probability of error, we now focus on a single event \(\mathcal{E}_{1}\).
One can verify that the analysis below is independent of the specific choice of \(i\) and applies uniformly to all error events \(\mathcal{E}_{i}\). Thus it also holds for \(\mathcal{E}_{\max}\).
Let \(V_j \in \{0,1\}\) denote the random variable that indicates if block \(B_j\) is infected (i.e., contains at least one infected item). 
The probability \(\qc\) 
that a given block is infected is
\begin{align}\label{def:qc}
    \qc
    &:= 1 - \Pr(U_1=0, U_2=0, \dots, U_C = 0) \nonumber\\
    &=
    1 - \left(1-q_n\middle)\middle(1-\alpha_n\right)^{C-1}.
\end{align}


Let \(\text{Pr}_{\rm FN}\) and \(\text{Pr}_{\rm FP}\) be the false negative and false positive probabilities for item~\(i\). Then
    \begin{align}\label{eq:errorfnfp}
        \Pr\left(\mathcal{E}_1\right)
        =
        q\text{Pr}_{\rm FN}
        +
        (1-q)\text{Pr}_{\rm FP}.
    \end{align}
A false negative error occurs when item~\(i\) is infected ($U_i=1$) but is incorrectly decoded as non-infected. This can only arise during the thresholding step, where \(\hat{u}_i = 0\) is mistakenly assigned.
This happens if the number of tests that contains the infected item falls below the threshold \(\gamma\). 
The false negative probability of error, denoted by \(\text{Pr}_{\rm FN}\), can be upper bounded. For \(\gamma < p \cdot T\), we have
\begin{align}\label{eq:FNupper}
    \text{Pr}_{\rm FN}
    =
    \Pr\left(\sum_{t=1}^{T} X_{t,1} \le \gamma \middle| U_1 = 1\right) \leq 2^{-T \cdot D(\gamma/T \| p)},
\end{align}
where \(D(a \| b)\) denotes the Kullback--Leibler (KL) divergence. This bound directly follows from the Chernoff bound. 
To see this note that although the testing design introduces correlations through block testing, each testing round is independent. 
This independence guarantees that the sequence of indicator random variables \(\{X_{t,1}\}_{t=1}^T\) are i.i.d.\ with distribution \(\text{Bern}(p)\), allowing us to apply Chernoff bound.
Setting 
\(\gamma = p(1 - \varepsilon) T\)
satisfies the condition for the above inequality. 

This justifies our initial choice of \(\gamma\) since this is the best threshold we can select to keep the false negative error low.
We now bound the probability of a false positive. 

A false positive error occurs when item~\(i\) is not infected but is decoded as infected, which requires: (1) the item passes the initial screening, i.e., in every test that includes the item, a infected item is also tested, and (2) it exceeds the threshold \(\gamma\) of tests required to set \(\hat u_i = 1\). Define the event
\begin{align*}
    \mathcal{E}_{1, \text{F.P.}}
    :=
    \left\{
      \tilde{u}_1 = 1, X_1 \geq \gamma
    \right\}.
\end{align*}

Recall that \(X_i = \sum_{t=1}^T X_{t,i}\) is the total number of tests in which item~\(i\) is included. For \(\gamma = p( 1 - \varepsilon)\,T\), by the law of total probability we have
\begin{align}\label{eq:FPupper}
    &\text{Pr}_{\rm FP}
    = 
    \Pr\left(\mathcal{E}_{1, \text{F.P.}} \middle| U_1=0\right) \nonumber \\
    &=
    \sum_{t=(p_1 - \varepsilon)T}^T
    \Pr\left(\tilde{u}_1 = 1 \middle|\, X_1 = t, U_1=0\right)
    \Pr\left(X_1=t\right) \nonumber \\
    &\stackrel{(a)}{\le}
    \Pr\left(\tilde{u}_1 = 1 \middle| X_1 = p(1 - \varepsilon)T,U_1=0\right),
\end{align}
where \((a)\) uses the first part of Lemma~\ref{lem:non-increasing} which states that \(\Pr\left(\tilde{u}_1 = 1 \middle| X_1 = t,U_1=0\right)\) is a non-increasing function of \(t\). 
 Intuitively, a non-infected item tested more frequently is less likely to be mistakenly identified as infected during the first decoding step, since it is more likely to appear in tests that are negative. 
We establish this in Lemma~\ref{lem:non-increasing}. This lemma is proved in Appendix~A of the longer version of this paper \cite{markovian_group_testing}.

\begin{restatable}{lemma}{FP}\label{lem:non-increasing}
For fixed \(p_1, p_2, C\), the function
\begin{align*}
    f(\gamma) 
    :=
    \Pr\left(\tilde{u}_1 = 1 \middle| X_1 = \gamma, U_1 = 0\right)
\end{align*}
is non-increasing with respect to \(\gamma\). Furthermore, \(f(\gamma)\) is upper bounded by
\begin{align*}
f(\gamma) \le \frac{1}{(1-q)^\gamma} \left[1 - (1-p_1)^{\qc\left(n/C-1\right)} (1-p_2)^{q(C-1)}\right]^\gamma,
\end{align*}
where \(\bar{p}_i := 1 - p_i\) for \(i = 1, 2\), and \(\qc\) is as defined in \eqref{def:qc}.
\end{restatable}

Applying Lemma~\ref{lem:non-increasing} with \(\gamma = p(1 - \varepsilon) T\), we upper bound \eqref{eq:FPupper} as
\begin{align}\label{eq:finalFPupper}
    \Pr\left(\tilde{u}_1 = 1 \middle| X_1 = p(1 - \epsilon) T, U_1 = 0\right)
    \le \left(r_n\right)^{p(1 - \varepsilon)T},
\end{align}
where
\[
r_n := \left[1 - (1-p_1)^{\qc(n/C-1)} (1-p_2)^{q(C-1)}\middle]\right/(1-q).
\]

To ensure a vanishing total probability of error, we analyze the minimum number of tests \(T\) required. 
From the union bound \eqref{eq:unionbound} and \eqref{eq:errorfnfp}, we have
\begin{align}\label{eq:unionerrorn}
\Pr(\mathcal{E}) \leq nq \text{Pr}_{\rm FN} + n(1 - q)\text{Pr}_{\rm FP}.
\end{align}

Using the bound for \(\Pr_{\rm FN}\) from \eqref{eq:FNupper} and the bound for \(\Pr_{\rm FP}\) from \eqref{eq:finalFPupper}, and substituting \(q = k\log n/n\), we obtain
\begin{align}\label{eq:finalerrorbd}
\Pr(\mathcal{E}) \leq (k\log n) 2^{-T \cdot D(p(1 - \varepsilon)\|p)} + n r_n^{p(1 - \varepsilon)T}.
\end{align}

To ensure that \(\Pr(\mathcal{E}) \to 0\) as \(n \to \infty\), we must appropriately choose the parameters \(p_1\), \(p_2\) and \(T\).
The following lemma establishes a testing rate \(\tau\) above which \eqref{eq:finalerrorbd} vanishes asymptotically. This lemma is proved in Appendix~B of the longer version of this paper \cite{markovian_group_testing}.

\begin{restatable}{lemma}{finaltest}\label{lem:finaltestbound}
Let the parameters \(p_1\), \(p_2\) and \(T\) be set as
\[
p_1 = \frac{\nu C}{n\qc}, \quad p_2 = 1 - 1/n, \quad T = \tau \cdot (nq_n\log{n}),
\] 
where \(n\) is to be large enough for \(p_1 < 1\). 
Then, for any \(\tau\) such that
\[
\tau > -\frac{\beta}{\nu \log(1 - \exp(-\nu))} + \frac{k^{\prime} - k}{C\nu \log(1 - \exp(-\nu))}
\]
the probability of error satisfies \(\Pr(\mathcal{E}) \to 0\) as \(n \to \infty\).
\end{restatable}
Lemma~\ref{lem:finaltestbound} implies that any
\[
\tau = -\frac{\beta}{\nu \log(1 - \exp(-\nu))} + \frac{k^{\prime} - k}{C\nu \log(1 - \exp(-\nu))},
\]
is achievable.
To obtain the sharp bound in Theorem~\ref{thm:achievability}, we optimize the term \(\nu \log(1 - \exp(-\nu))\). 
It can be shown that this expression achieves its maximum at \(\nu = \ln 2\). 
Substituting this value of \(\nu\), and letting \(C \to \infty\) the bound simplifies to
\[
\tau = \frac{\beta}{\ln{2}},
\]
is achievable. Therefore the minimal testing rate \(\tau^*\) satisfies
\[
\tau^* \leq \frac{\beta}{\ln{2}} \approx 1.44\beta.
\]
This proves the result in Theorem~\ref{thm:achievability}.

\section{Conclusion}\label{sec:Conclusion}
In this work, we introduced a novel Markovian correlation model for correlated group testing.
We proposed a novel \emph{per-item} testing and decoding strategy for non-adaptive group testing in the presence of correlated infection patterns. 

Future work could extend this framework to more general correlation structures, such as higher-order Markov chains or graphical models, to capture long-range dependencies and complex item relationships. Exploring joint decoding strategies, which consider the infection status of all items simultaneously, could further reduce the number of required tests. Additionally, adapting this approach to the sparse regime, where infections scale sub-linearly with population size \((nq_n \sim O(n^\theta), \theta < 1)\), is another important direction to consider.


\newpage

{\footnotesize
\bibliographystyle{ieeetr}
\bibliography{Ref.bib}
}

\newpage
\appendices
\section{Proof of Lemma~\ref{lem:non-increasing}}
\FP*

\begin{proof}
We first establish that \(f(\gamma)\) is non-increasing with respect to \(\gamma\). By definition,
\begin{align*}
    f(\gamma) 
    :=
    \Pr\left(\tilde{u}_1 = 1 \middle| X_1 = \gamma, U_1 = 0\right).
\end{align*}
Recall that a false positive occurs when an uninfected item is incorrectly identified as infected. 
This error can only happen if all tests involving item \(1\) return a positive outcome. 
Conditioning on \(X_1 = \gamma\), meaning item \(1\) participates in exactly \(\gamma\) tests, we can rewrite \(f(\gamma)\) as
\begin{align*}
    f(\gamma) 
    &= \Pr\left(Y_t = 1 \,\, \forall\, t \text{ such that } X_{t,1} = 1 \middle| X_1 = \gamma, U_1 = 0\right).
\end{align*}

Due to the independence of test outcomes across different tests, we can further simplify \(f(\gamma)\) as
\begin{align*}
    f(\gamma) 
    &= \prod_{t: X_{t,1} = 1} \Pr\left(Y_t = 1 \middle|X_{t,1} = 1, U_1 = 0\right).
\end{align*}

Since the probability \(\Pr\left(Y_t = 1 \middle|X_{t,1} = 1, U_1 = 0\right)\) is identical across all tests that include item \(1\), we denote this common probability as \(p_{\text{FP}} := \Pr\left(Y_1 = 1 \middle| X_{1,1} = 1, U_1 = 0\right)\). 
Thus, we have
\begin{align}\label{eq:pfp}
    f(\gamma) = \left(p_{\text{FP}}\right)^\gamma.
\end{align}
Observe that \(0 \leq p_{\text{FP}} \leq 1\). Therefore, \(f(\gamma)\) is non-increasing in \(\gamma\). This completes the proof of the first part of the lemma.

The second part of the lemma, i.e., establishing the upper bound, is more involved. Recall that the infection vector is
\[
U^n = (U_1, U_2, \dots, U_n).
\]
Further define the block infection vector as
\[
V^n = (V_1, V_2, \dots, V_{n/C}),
\]
where each random variable \(V_j \in \{0,1\}\) indicates whether at least one item in block \(j\) is infected. Note that \(V^n\) is a deterministic function of \(U^n\). 
We begin with 
\[
p_{\mathrm{FP}} = \Pr\left(Y_1 = 1 \middle| X_{1,1} = 1, U_1 = 0\right).
\]
We apply the law of total probability with respect to the random vector \(U^n\), which captures the infection status of all the blocks:
\begin{align*}
p_{\rm FP} &= \sum_{u^n} \Pr\left(Y_1 = 1\middle| X_{1,1} = 1, U_1 = 0, u^n\right) \Pr\left(u^n \middle| U_1 = 0\right) \\
&= \mathbb{E}_{U^n}\left[\Pr\left(Y_1 = 1 \middle| X_{1,1} = 1, U_1 = 0, U^n\right) \middle| U_1 = 0\right].
\end{align*}
To simplify notation, define the function
\[
g(U^n) := \Pr\left(Y_1 = 1 \middle| X_{1,1} = 1, U_1 = 0, U^n\right).
\]
Hence
\begin{align} \label{eq:g_cond}
\Pr\left(Y_1 = 1 \middle| X_{1,1} = 1, U_1 = 0\right) = \mathbb{E}_{U_2^n}\left[g(U^n) \middle| U_1 = 0\right],
\end{align}
where $U_2^n = (U_2,\dots,U_n)$.
We now remove the conditioning on \(U_1 = 0\) to make the analysis simpler. Using the law of total probability, we can write the unconditional expectation \(E[g(U^n)]\) as
\[
\mathbb{E}\left[g(U^n)\right] = q \mathbb{E}\left[g(U^n) \middle| U_1 = 1\right] + (1-q) \mathbb{E}\left[g(U^n) \middle| U_1 = 0\right],
\]
where \(q \) denotes the probability that the first item is infected. Rearranging the terms and using the fact that  \(g(\cdot) \geq 0\), we obtain
\[
\mathbb{E}\left[g(U^n) \middle| U_1 = 0\right] \le \frac{1}{1-q} \mathbb{E}\left[g(U^n)\right].
\]
Hence,
\[
\Pr\left(Y_1 = 1 \middle| X_{1,1} = 1, U_1 = 0\right) \le \frac{1}{1-q} \mathbb{E}\left[g(U^n)\right].
\]

We now look closely at 
\[
g(U^n) = 1 - \Pr\left(Y_1 = 0 \middle| X_{1,1} = 1, U_1 = 0, U^n\right).
\]
In particular, \(Y_1 = 0\) occurs if no infected item in the pool contributes to a positive test outcome. 
For each block that item \(1\) does not belong to (namely \(B_2, B_3, \dots, B_{C}\)), the event \(\{Y_1 = 0\}\) can happen in one of the following scenarios:
\begin{enumerate}
    \item The block is not infected. (In this case, it cannot contribute to a positive test).
    \item The block is infected but is not selected for testing, which happens with probability \((1-p_1)\).
    \item The block is infected and is selected for testing (which occurs with probability \(p_1\)), but none of its infected items are ultimately tested. This last scenario happens with probability \(\left(1 - p_2\right)^{\sum_{j=1}^{C} U_j'}\), where \(U_j'\) indicates whether a particular item in that block is infected, given that the block is infected.
\end{enumerate}
Hence, if a block is infected, the probability it does not contribute to a positive test is
\[
(1-p_1) + p_1\left(1-p_2\right)^{\sum_{j=1}^{C} U_j'} \ge (1-p_1).
\]
Discarding the second term in the above equation helps simplify the analsyis further.  

Suppose there are \(\sum_{j=2}^{n/C} V_j\) infected blocks among those that do not contain item 1. 
The probability that all those infected blocks fail to produce a positive test is (recall that \(V_j\) is a deterministic function of \(U_j\))
\[
\left((1-p_1) + p_1\left(1-p_2\right)^{\sum_{j=1}^{C} U_j'}\right)^{\sum_{j=2}^{n/C} V_j} \ge (1-p_1)^{\sum_{j=2}^{n/C} V_j}.
\]

Within the block containing item 1, the probability that none of its other infected items are included in the test becomes
\[
\left(1 - p_2\right)^{\sum_{j=2}^{C} U_j}.
\]
Putting these pieces together, \(g(U^n)\) can be upper bounded by
\begin{align*}
&g(U^n) \leq 1 - (1-p_1)^{\sum_{j=2}^{n/C} V_j}\left(1-p_2\right)^{\sum_{j=2}^{C} U_j}.
\end{align*}
Collecting all steps, we arrive at the bound
\[
\mathbb{E}[g(U^n)] \le \frac{1}{1-q}\mathbb{E}\left[1 - (1-p_1)^{\sum_{j=2}^{n/C} V_j}(1-p_2)^{\sum_{j=2}^{C} U_j}\right].
\]
Now define
\[
A :=\left(\sum_{j=2}^{n/C} V_j\right) \log{(1-p_1)} + \left(\sum_{j=2}^{C} U_j\right)\log{(1-p_2)}
\]
and
\begin{align*}
&\Phi(A) :=1 - 2^A.
\end{align*}
Now we claim \(\Phi(A)\) is concave in \(A\). To verify this, note that \(2^A\) is an exponential function and is convex in \(A\) for all \(A\). 
Since \(\Phi(A)\) is defined as \(1 - 2^A\), the result is a concave function.
Applying Jensen's inequality to the concave function \(\Phi(A)\), we have
\[
\mathbb{E}\left[\Phi(A)\right] \le \Phi\left(\mathbb{E}[A]\right).
\]
Now, note that 
\[
\mathbb{E}[A] = \qc\left(\frac{n}{C}-1\right)\log{(1-p_1)} + q(C-1)\log{(1-p_2)},
\] 
where we recall that \(\qc = \Pr(V_j = 1)\) is the probability that a block \(j\) is infected. Substituting this into the inequality gives
\begin{align*}
&\mathbb{E}\left[1 - 2^A\right] 
\le 1 - (1-p_1)^{\qc\left(\frac{n}{C}-1\right)} (1-p_2)^{ q(C-1)}.
\end{align*}
Combining this with the earlier inequality, we obtain
\[
g(U^n) \le \frac{1}{1-q} \left[1 - (1-p_1)^{\qc\left(\frac{n}{C}-1\right)} (1-p_2)^{ q(C-1)}\right].
\]
Finally, recall that from (\ref{eq:pfp}) and \eqref{eq:g_cond},

\[
f(\gamma) = \mathbb{E}\left[g(U^n) \middle| U_1 = 0\right]^\gamma.
\]
Substituting the bound for \( \mathbb{E}\left[g(U^n) \middle| U_1 = 0\right]\), we find
\[
f(\gamma) \le \left(\frac{1}{1-q} \left[1 - (1-p_1)^{\qc\left(\frac{n}{C}-1\right)} (1-p_2)^{q(C-1)}\right]\right)^\gamma.
\]
Expanding this gives the final result
\[
f(\gamma) \le \frac{1}{(1-q)^\gamma} \left[1 - (1-p_1)^{\qc\left(n/C-1\right)} (1-p_2)^{q(C-1)}\right]^\gamma.
\]

\end{proof}
\section{Proof of Lemma~\ref{lem:finaltestbound}}
\finaltest*

\begin{proof}
From \eqref{eq:unionerrorn}, the total probability of error \(\Pr(\mathcal{E})\) is bounded as
\begin{align*}
\Pr(\mathcal{E}) &\leq (k \log n) 2^{-T \cdot D(p(1 - \varepsilon) \| p)} + n r_n^{p(1 - \varepsilon) T}  \\
&= (k \log n) 2^{-T \cdot D(p(1 - \varepsilon) \| p)} + 2^{\log{n} + p(1 - \varepsilon) T\log{r_n}},
\end{align*}
where \(r_n\) is defined as
\[
r_n = \frac{1 - (1 - p_1)^{\qc (n / C - 1)} (1 - p_2)^{q(C-1)}}{1 - q}.
\]

We first focus on the second term \(n r_n^{p(1 - \varepsilon) T}\). 
For this term to vanish as \(n \to \infty\), we can set
\[
T = \frac{-(1 + \delta) \log n}{p(1 - \varepsilon) \log r_n},
\]
for some \(\delta > 0\). We now focus on the limit 
\[\lim_{n \to \infty} T(n)/ (nq_n\log{n}) = \lim_{n \to \infty} T(n)/ (k\log^2{n})\]
to show the required bound. Note that 
\[
\lim_{n \to \infty} T(n)/ k\log^2{n} =-\frac{1+\delta}{k(1-\varepsilon)}\left( \lim_{n \to \infty} \frac{ 1}{p\log{n}}\right) \left(\lim_{n \to \infty} \frac{1}{\log r_n}\right)
\]
By definition, \(\qc := 1- (1-q)(1-\alpha_n)^{C - 1}\), \(p_1 = \frac{\nu C}{n \qc}\) and \(p_2 = 1 - 1/n\). Substituting this in the first limit after the constant factor we get
\[
 \frac{ 1}{p\log{n}} = \frac{n \cdot (1- (1-q)(1-\alpha_n)^{C - 1})}{\nu C (1 - 1/n) \log{n}}.
\]
Recall that \(\alpha_n = \frac{k' \log n}{n}\) and \(q = \frac{k \log n}{n}\). We now evaluate the limit
\begin{align}\label{eq:first}
&\lim_{n \to \infty}\frac{n(1 - (1-q)(1 - \alpha_n)^{C - 1})}{C(1 - 1/n)\log{n}} \nonumber \\
&\stackrel{(a)}{=}  \lim_{n \to \infty}\frac{n(1 - (1-q)(1 - (C-1)\alpha_n + \binom{C-1}{2}\alpha_n^2 \dots))}{C\log{n}} \nonumber \\
&\stackrel{(b)}{=} \lim_{n \to \infty} \frac{n((C-1)\alpha_n - q_n + o(\alpha_n))}{C\log{n}}\nonumber \\
&= \lim_{n \to \infty} \frac{n\alpha_n}{\log{n}} +  \lim_{n \to \infty} \frac{nq_n - n\alpha_n}{C\log{n}}\nonumber \\
&= k^{\prime} + \frac{k - k^{\prime}}{C},
\end{align}
where (a) is obtained by applying Taylor's expansion on \((1 - \alpha_n)^{C - 1}\) and noting that \(1/n \to 0\) as \(n \to \infty\), and (b) is due to the fact that \(\alpha_n C \sim \log{n}/n\), which means that higher powers of the Taylor's  series reduce in order.

We now consider the second limit
\(
\lim_{n \to \infty} \frac{1}{\log r_n}.
\)
Examine the limit of the numerator of \(r_n\) which is
\[1 - (1 - p_1)^{\qc (n / C - 1)} (1 - p_2)^{q(C-1)}.\]
First, note that as \(n \to \infty\),
\[
(1 - p_1)^{\qc (n / C - 1)} \sim \exp(-(\nu C/ n \qc)\times \qc(n/C)) \to \exp{(-\nu)},\]
and
\[(1 - p_2)^{q(C-1)} = \left(\frac{1}{n}\right)^{(C-1)k\log{n}/n} \sim \exp(-\log^2{n}/n)\to 1.
\]
Thus,
\[
1 - (1 - p_1)^{n \qc / C_n} (1 - p_2)^{C_n} \to 1 - e^{-\nu}.
\]
For the denominator, note that \(q = \frac{k \log n}{n}\), so \(1 - q \to 1\) as \(n \to \infty\). Therefore,
\[
r_n \to 1 - e^{-\nu}.
\]
Taking the logarithm,
\begin{align}\label{eq:secondlimit}
\log r_n \to \log(1 - e^{-\nu}).
\end{align}
Combining the limits \eqref{eq:first} and \eqref{eq:secondlimit},
we conclude that
\begin{align}\label{eq:limittests}
\lim_{n \to \infty } T(n)/k\log^2{n} =  -\frac{(1+\delta)(k^{\prime} + (k - k^{\prime})/C)}{k(1-\varepsilon)\nu \log(1 - \exp(-\nu))}.
\end{align}

Now we consider the first term in \eqref{eq:unionerrorn}. To show that 
\[
(k \log n)\,2^{-T \cdot D(p(1-\varepsilon) \| p)} \to 0,
\]
we start by stating the following bound
\begin{align}\label{eq:KLinequality}
(k \log n)\,2^{-T \cdot D(p(1-\varepsilon) \| p)} \leq (k \log n)\,2^{-T \cdot p\varepsilon^2/2}.
\end{align}

This inequality follows from the bound (when \(x < y\))
\begin{align}\label{eq:KLinequalityxy}
D(x \| y) \geq \frac{(x-y)^2}{2y},
\end{align}
where \(x =p(1-\epsilon)\) and \(y = p\). To prove \eqref{eq:KLinequalityxy}, expand \(D(x \| y)\) via a Taylor series around \(x = y\) and treat it as a function of \(x\). Let \(h(x) := D(x \| y)\). Then
\[
\frac{\partial}{\partial x} h(x) = \ln{2}\left(\log \frac{x}{y} - \log \frac{1-x}{1-y}\right),
\]
\[
\frac{\partial^2}{\partial x^2}h(x) = (\ln{2})^2\left(\frac{1}{x} + \frac{1}{1-x}\right).
\]
Using Taylor's expansion, expand \(h(x)\) around \(x = y\)
\[
h(x) = h(y) + \frac{\partial}{\partial x} h(y) \cdot (x-y) + \frac{\partial^2}{\partial x^2} h(\xi) \cdot \frac{(x-y)^2}{2},
\]
where \(\xi \in [x, y]\).
Since \(h(y) = 0\) and \(\frac{\partial}{\partial x} h(y) = 0\), we have
\[
h(x) = \frac{\partial^2}{\partial x^2} h(\xi) \cdot \frac{(x-y)^2}{2}.
\]
For \(0 \leq \xi \leq y \leq 1\), it follows that
\[
\frac{\partial^2}{\partial x^2} h(\xi) = \frac{1}{\xi} + \frac{1}{1-\xi} \geq \frac{1}{\xi} \geq \frac{1}{y}.
\]
Thus, when \(x > y\)
\[
D(x \| y) \geq \frac{1}{2} \cdot \frac{1}{y} \cdot (x-y)^2 = \frac{(x-y)^2}{2y}.
\]
Now substituting \(x = p(1-\varepsilon)\) and \(y = p\), we obtain \eqref{eq:KLinequality}.

Now in \eqref{eq:KLinequality} note that \(k \log n\) grows logarithmically in \(n\) and the exponential term 
decays as \(n \to \infty\).

From the parameter choices
\[
T \sim O(\log^2 n), \quad p = \frac{\nu C}{nq_n} \sim O\left(\frac{1}{\log n}\right),
\]
it follows that
\[
T \cdot p \epsilon^2 \sim O(\log n).
\]

Thus, we can conclude that
\begin{align}\label{eq:firsttermlimit}
(k \log n) \cdot 2^{-T \cdot D(p(1-\varepsilon) \| p)} \sim O\left(\frac{k \log n}{n}\right) \to 0
\end{align}
as \(n \to \infty\).

From \eqref{eq:limittests} and \eqref{eq:firsttermlimit}, we have shown when
\[
\lim_{n \to \infty} \frac{T(n)}{nq_n\log{n}} = -\frac{(1+\delta)(k^{\prime} + (k - k^{\prime})/C)}{k(1-\varepsilon)\nu \log(1 - \exp(-\nu))},
\]

the total probability of error from \(\Pr(\mathcal{E}) \to 0\) as \(n \to \infty\)

Recall from Section~\ref{sec:ProblemSetting} that \(k^{\prime} = \beta n q_n\), 
we can further claim (letting \(\delta, \varepsilon \to 0\)) that
\[
\tau^ > -\frac{\beta}{\nu \log(1 - \exp(-\nu))} + \frac{k- k^{\prime}}{C\nu \log(1 - \exp(-\nu))} 
\]
suffices for \(\Pr(\mathcal{E}) \to 0\) as \(n \to \infty\).
\end{proof}
\section{Proof of Theorem~\ref{thm:converse}}

\begin{proof}
We begin by considering a lower bound on the minimal group testing rate, as derived in \cite{wolf1985born}. The result is given by:
\begin{align}\label{eq:taustar}
    \tau^{*} > \lim_{n \to \infty} \frac{1}{n q_n \log n} H(U^n),
\end{align}
where \( q_n = \frac{k \log n}{n} \) represents the infection rate, and \( H(U^n) \) is the entropy of the infection pattern \( U^n = (U_1, U_2, \ldots, U_n) \).

Under the Markovian correlation model parameterized by \(\alpha\) and \(\beta\), the entropy \(H(U^n)\) can be expanded using the chain rule for entropy:
\begin{align}
    H(U^n) &= H(U_1) + \sum_{i=2}^n H(U_i | U_{i-1}).
\end{align}

Since the Markov chain is stationary, the conditional entropy \(H(U_i | U_{i-1})\) is the same for all \(i \geq 2\), and equals \(H(U_2 | U_1)\). 
Therefore, the total entropy simplifies to:
\begin{align}
    H(U^n) &= H(U_1) + (n-1)H(U_2|U_1).
\end{align}

Next, we expand \(H(U_2|U_1)\) using the Markov transition probabilities. By definition:
\begin{align}
    H(U_2|U_1) &= q H(U_2 | U_1 = 1) + (1-q) H(U_2 | U_1 = 0),
\end{align}
where \(q\) is the stationary probability of \(U_1 = 1\). 

For the conditional entropy \(H(U_2 | U_1 = 1)\), note that \(U_2 | U_1 = 1 \sim \text{Bern}(1-\beta)\). 
Hence,
\begin{align}
    H(U_2 | U_1 = 1) = -\beta \log \beta - (1-\beta) \log (1-\beta).
\end{align}
Similarly, for \(H(U_2 | U_1 = 0)\), we have \(U_2 | U_1 = 0 \sim \text{Bern}(\alpha)\). Thus,
\begin{align}
    H(U_2 | U_1 = 0) = -\alpha \log \alpha - (1-\alpha) \log (1-\alpha).
\end{align}
Substituting these results into the expression for \(H(U_2|U_1)\), we obtain
\begin{align}
    H(U_2|U_1) &= q \left(-\beta \log \beta - (1-\beta) \log (1-\beta)\right) \nonumber \\
    &\quad + (1-q) \left(-\alpha \log \alpha - (1-\alpha) \log (1-\alpha)\right).
\end{align}
For large \(n\), the infection rate \(\alpha = \frac{k^{\prime} \log n}{n}\) dominates due to its asymptotic behavior. Specifically, the leading term of \(H(U_2|U_1)\) is given by
\begin{align}
    H(U_2|U_1) &= \alpha \log \frac{1}{\alpha} + o\left(\alpha \log \frac{1}{\alpha}\right).
\end{align}
Substituting \(\alpha = \frac{k^{\prime} \log n}{n}\), we obtain
\begin{align}
    H(U_2|U_1) &= \frac{k^{\prime} \log n}{n} \log \frac{n}{k^{\prime} \log n} + o\left(\frac{\log^2 n}{n}\right) \nonumber \\
    &= \frac{k^{\prime} \log^2 n}{n} + o\left(\frac{\log^2 n}{n}\right).
\end{align}
Thus, for large \(n\), the entropy \(H(U^n)\) is approximated as
\begin{align}
    H(U^n) &= n \cdot \frac{k^{\prime} \log^2 n}{n} + o(\log^2 n) \nonumber \\
    &= k^{\prime} \log^2 n + o(\log^2 n).
\end{align}
Returning to the bound in \eqref{eq:taustar}, we substitute \(H(U^n)\) and \(n q_n = k \log n\), where \(k = \frac{k^{\prime}}{\beta}\), to obtain:
\begin{align}
    \tau^{*} &> \lim_{n \to \infty} \frac{k^{\prime} \log^2 n + o(\log^2 n)}{k \log^2 n} \nonumber \\
    &= \beta.
\end{align}
This completes the proof.
\end{proof}
\end{document}